\documentclass{article}%
\usepackage{amsfonts}
\usepackage{amsmath}
\usepackage{amssymb}
\usepackage{graphicx}%
\providecommand{\U}[1]{\protect\rule{.1in}{.1in}}
\newtheorem{theorem}{Theorem}
\newtheorem{acknowledgement}[theorem]{Acknowledgement}

\newtheorem{conclusion}[theorem]{Conclusion}

\newtheorem{definition}[theorem]{Definition}

\newtheorem{proposition}[theorem]{Proposition}
\newtheorem{remark}[theorem]{Remark}

\newenvironment{proof}[1][Proof]{\noindent\textbf{#1.} }{\ \rule{0.5em}{0.5em}}
\begin{document}

\title{Category theoretic properties of the A. R\'{e}nyi and C. Tsallis entropies.}
\author{Gy\"{o}rgy Steinbrecher\\Physics Department, University of Craiova, A. I. Cuza 13, 200585\\Craiova, Romania. Email: gyorgy.steinbrecher@gmail.com
\and Alberto Sonnino\\Karlsruhe Institute of Technologies (KIT)\\Department of Electrical Engineering and Information Technologies\\D-76131 Karlsruhe, Germany \&\\Ecole Polytechnique de Louvain (EPL)\\Universit\'{e} Catholique de Louvain (UCL)\\Rue Archimede 1 bte L6.11.01, 1348 Louvain-la-Neuve - Belgium. \\Email: alberto.sonnino@gmail.com
\and Giorgio Sonnino\\Department of Theoretical Physics and Mathematics\\Universit\'{e} Libre de Bruxelles (ULB)\\Campus Plaine CP 231, Bvd de Triomphe, 1050 Brussels, Belgium.\\\& Royal Military School (RMS)\\Av. de la Renaissance 30 1000 Brussels - Belgium\\Email: gsonnino@ulb.ac.be}
\maketitle

\begin{abstract}
The problem of embedding the Tsallis and R\'{e}nyi entropies in the framework
of category theory and their axiomatic foundation is studied. To this end, we
construct a special category MES related to measured spaces. We prove that
both of the R\'{e}nyi and Tsallis entropies can be imbedded in the formalism
of category theory by proving that the same basic functional that appears in
their definitions, as well as in the associated Lebesgue space norms, has good
algebraic compatibility properties. We prove that this functional is both
additive and multiplicative with respect to the direct product and the
disjoint sum (the coproduct) in the category MES, so it is a natural candidate
for the measure of information or uncertainty. We prove that the category MES
can be extended to monoidal category, both with respect to the direct product
as well as to the coproduct. The basic axioms of the original R\'{e}nyi
entropy theory are generalized and reformulated in the framework of category
MES and we prove that these axioms foresee the existence of an universal
exponent having the same values for all the objects of the category MES. In
addition, this universal exponent is the parameter, which appears in the
definition of the Tsallis and R\'{e}nyi entropies.

\end{abstract}

\section{Introduction}

The discovery of two related generalizations of the classical Shannon entropy
\cite{shannon} is a remarkable coincidence in the history of abstract
probability theory and statistical physics. A. R\'{e}nyi introduced a possible
generalization \cite{Renyi1} of the classical Shannon entropy by pure
axiomatic extension of the Fadeev axioms \cite{FadeevAxioms}, \cite{Feinstein}
that define uniquely the Shannon entropy. On the other hand, the generalized
entropy \cite{Tsallis1}, \cite{Tsallis2} introduced by C. Tsallis was useful
to extend the classical maximum entropy principle such that the heavy tailed
distributions observed in a large scale of physical processes
\cite{TsallisBook}, \cite{TsallisGelMann}, \cite{SGBW}, \cite{SGW}, could be
derived from (generalized) maximum entropy principles. The interest in the
study of the generalizations of the Shannon entropy in the recent years is due
to the multiple applications of the Tsallis and R\'{e}nyi entropy or the
associated R\'{e}nyi divergence \cite{TsallisBook}, \cite{TsallisGelMann},
\cite{maszczyk}, \cite{RenyiDivergence}. We mention also that similar to the
classical \textit{H theorem of L. Boltzmann}, the generalizations of the
R\'{e}nyi entropy, as well as the original R\'{e}nyi entropy, is a Liapunov
functional for a large class of stochastic processes described by generalized
Fokker-Planck equations, more exactly by Fokker-Planck equation where the
drift term and the diffusion tensor are itself dependent on some external
random variable \cite{SonninoSteinbrGRE}. We mention that in the case of
suitable singular limiting procedure, both the Tsallis and R\'{e}nyi entropies
give the same limit: the Shannon entropy. The R\'{e}nyi entropy is additive
while the Tsallis entropy is not. Despite the R\'{e}nyi and Tsallis entropies
give the same results in the case of problems associated to the determination
of the probability density function from the Maximum Entropy principles,
because they are algebraically related by simple formulae, the non-additivity
of the Tsallis entropy generated many discussions in the physical literature.
On the other hand, by formulating the basic axioms \cite{Renyi1}, A. R\'{e}nyi
introduced new concepts (\textit{incomplete random variables and incomplete
distributions}) that are not included in the standard terminology of the
probability theory. Also the formulation of the \textit{Postulate} 5'
\cite{Renyi1}, is not the simplest, mathematically natural.

\noindent In this work we develop a formalism in the framework of the category
theory \cite{Categories} and \cite{MacLaneCategory} for the study of
generalized entropies. The category theory is the branch of mathematics that
plays a central role in the logical foundation and synthesis of the whole
contemporary mathematics. In particular, the category theory allows avoiding
the paradoxes of the classical set theory. In order to highlight the natural
structures related to generalized entropies, we use the central concepts of
the modern mathematics.

\noindent The paper is organized as follows. In the Section
\ref{markerSectionFunctorialProperties},
Subsection~\ref{markerSubsectDefinitions}, we define a special category
related to measurable spaces (referred to as $MES$), enabling the introduction
an \textit{associated basic functional} $Z_{p}$ (see the forthcoming Section
for his exact definition). Both the Tsallis and R\'{e}nyi entropies, as well
as the distance in $L^{p}$ spaces, may be expressed in terms of this
functional. In the Subsection~\ref{markerSubsectDirectProduct}, we define the
direct product of the objects in $MES$ and we prove that the functional
$Z_{p}$ satisfies a \textit{compatibility relation} with respect to this
product i.e., it is \textit{multiplicative}. This multiplicative property is
equivalent to the additivity of the R\'{e}nyi entropy. In the
Subsection~\ref{markerSubsectCoproduct} we define the disjoint sum (or the
coproduct) of the objects in $MES$, and we prove that the functional $Z_{p}$
satisfies a \textit{compatibility relation} with respect to coproduct i.e., it
is \textit{additive}. Note that this property is equivalent to one of the
postulates characterizing the R\'{e}nyi entropy. The prove that both product
and coproduct possess a universal property and that the direct product and
coproduct can also be defined for morphisms of the category $MES$, can be
found in the Subsection~\ref{markerSubsectUniversalProperties}. In the
Subsection~\ref{markerSubsectMonoidalCategory} we show that, by extending the
category $MES$ with the introduction of the \textit{unit object} and the null
\textit{object}, the category $MES$ \ became to a \textit{monoidal category}.

\noindent Section~\ref{markerSectionAxioms} deals with the axiomatic
characterization of the functional $Z_{p}$. We demonstrate that there exists a
universal exponent $p$ (the same for all the objects of the category) that
characterizes completely the functional $Z_{p}$ (hence, also the Tsallis or
R\'{e}nyi entropies) up to an arbitrary multiplicative factor.

\noindent Appendix~\ref{markerAppendxRenyDivergenceEntropy} shows that the
R\'{e}nyi divergence can be expressed in terms of the R\'{e}nyi entropy. The
proof of the universality (with respect to all the objects of the category
$MES$) of the exponent defining the R\'{e}nyi or Tsallis entropies can be
found in Appendix~\ref{markerAppendixSubsectFunctionalEquaton}.

\section{The category-theoretic properties related to R\'{e}nyi and Tsallis
entropies. \label{markerSectionFunctorialProperties}}

\subsection{Definitions \label{markerSubsectDefinitions}}

Our definitions include as a particular case the original definition of the
generalized entropies \cite{Tsallis1}, \cite{Tsallis2} and \cite{Renyi1}. Our
basic construction that will play the role of the \emph{object of the category
MES} is derived from the well known concept of measurable space \cite{Rudin},
\cite{ReedSimon}. Guided by statistical ideas, in order to take into account
the \textit{negligible sets} we specify also an sub-ideal of the $\sigma
$-algebra of measurable sets. \emph{The objects of the category }$MES$ consist
of triplets $M_{X}:=(X,\mathcal{A}_{X},\mathcal{N}_{X})$ with $X$ denoting the
phase space (for instance, it is a symplectic manifold in the case of
statistical physics or, in the case of elementary probability models, finite
or denumerable set) and $\mathcal{A}_{X}$ is the $\sigma-$algebra generated by
a family of subsets of $X$, respectively. We also denote with $\mathcal{N}%
_{X}\subset\mathcal{A}_{X}$ an ideal of the $\sigma$-algebra $\mathcal{A}_{X}$
having the meaning of \textit{negligible sets}. Let us now postulate the
\textit{completeness property}. From $N\in\mathcal{N}_{X}$ and $N^{\prime
}\subset N$ results $N^{\prime}\in$ $\mathcal{N}_{X}$. \emph{The morphisms of
the category }$MES$ with the source $M_{X}$ and range $M_{Y}$ are the
measurable maps $\Phi$ from $X~$\ to $Y$, which are \textit{nonsingular} i.e.,
such that $\Phi^{-1}(\mathcal{N}_{Y})\subset\mathcal{N}_{X}$. From the
completeness property results the \textit{ideal property} i.e., if
$N\in\mathcal{N}_{X}$ and $A\in\mathcal{A}_{X}$ then $A\cap N\in
\mathcal{N}_{X}$. Note that it is possible that $\mathcal{N}_{X}$ contains
only the empty set (as, for example, in the case of atomic spaces).

\begin{remark}
At first sight it would be more natural to consider the objects as
\textit{measure space triplet} $(X,\mathcal{A}_{X},\mu_{X})$ containing the
measure $\mu_{X}$, and the morphisms as the measure preserving
transformations. However, in this case we cannot define direct product or
coproduct having universal property.
\end{remark}

\noindent We denote with $C(M_{X})$, or with $C(X,\mathcal{A}_{X}
,\mathcal{N}_{X})$, the cone with all $\sigma-$ finite positive measures over
$(X,\mathcal{A}_{X},\mathcal{N}_{X})$ that are compatible with $\mathcal{N}
_{X}$ (i.e., $\mu\in C(X,\mathcal{A}_{X},\mathcal{N}_{X})$ iff for all
$N\in\mathcal{N}_{X}$ we have $\mu(N)=0$). For a given $\mu_{X}\in
C(X,\mathcal{A}_{X},\mathcal{N}_{X})$ and $p>0$, we denote with $L^{p}
(M_{X},\mu_{X})$ the Banach space ($p\geq1$) or the Fr\'{e}chet space
($0<p<1$) of functions $f_{X}:$ $X\rightarrow\mathbb{R}$ that are measurable
modulo $\mathcal{N}_{X}$ and have finite norm (pseudo norm, respectively):
more precisely, ${\int\limits_{X}}\left\vert f_{X}(x)\right\vert ^{p}d\mu
_{X}(x\mathbf{)<\infty}$. In the sequel, we shall denote
\begin{equation}
Z_{p}(M_{X},\mu_{X},{\rho}_{X}\text{ }):={\int\limits_{X}}\rho_{X}(x)^{p}
d\mu_{X}(x\mathbf{)} \label{LL0.01}%
\end{equation}
\noindent for some non-negative density $\rho_{X}\in L^{p}(M_{X},\mu_{X})$.
The generalized entropies are defined for probability density functions (PDF)
satisfying the conditions
\begin{align}
&  \rho_{X}\in L^{1}(M_{X},\mu_{X})\cap L^{p}(M_{X},\mu_{X});~\ \label{LL0.1}%
\\
&  \int\limits_{X}\rho_{X}(x)d\mu_{X}(x)=1 \label{LL0.1.1}%
\end{align}
\noindent where $\ p>0$ and $p\neq1$. The probability $P(A)$ can be
represented by PDF as follows
\begin{align}
P(A)  &  =\int\limits_{A}\rho_{X}(x)d\mu_{X}(x);\label{LL1}\\
~A  &  \subset X;~A\in\mathcal{A}_{X};~\mu_{X}\in C(M_{X})
\end{align}
\noindent In this framework, for a given measurable space $M_{X}
:=(X,\mathcal{A}_{X},\mathcal{N}_{X})$ and measure $\mu_{X}\in C(M_{X})$, the
classical Boltzmann-Gibbs-Shannon entropy functional is given by
\begin{equation}
S_{cl}[M_{X},\mu_{X},\rho_{X}]=-{\int\limits_{X}}\rho_{X}(x)\log\left[
\rho_{X}(x)\right]  d\mu_{X}(x) \label{LL2}%
\end{equation}
\noindent For a given measurable space $M_{X}$, the generalizations of the A.
R\'{e}nyi \cite{Renyi1} and C. Tsallis \cite{Tsallis1}, \cite{Tsallis2}
entropies, involves the functional $Z_{p}(M_{X},\mu_{X},{\rho}_{X})$ given by
Eq.(\ref{LL0.01}). The functional $Z_{p}$ is related to the norm of the
density $\rho$ in the Banach space for $p\geq1$ \cite{ReedSimon}, and to the
pseudo-norm $N_{p}[\rho]$ for $0<p<1$ \cite{Rudin} \cite{LuschgiPages},
through the obvious relations
\begin{align}
\left\Vert \rho_{X}\right\Vert _{p}  &  =\left[  {\int\limits_{X}}\left[
\rho_{X}(x)\right]  ^{p}d\mu_{X}(x)\right]  ^{\frac{1}{p}};~p\geq
1\label{LL3}\\
N_{p}[\rho_{X}]  &  ={\int\limits_{\Omega}}\left[  \rho_{X}(x)\right]
^{p}d\mu_{X}(x);~0<p\leq1 \label{LL4}%
\end{align}
\noindent These relations give the geometrical interpretation of the
generalized entropies (for further information Refs to
\cite{SonninoSteinbrGRE}).

\begin{remark}
The study of the generalized entropies helps us to better understand the
classical entropy. For $p\geq1$, the functional $\left\Vert \rho
_{X}\right\Vert _{p}$ is the classical $L^{p}$ norm, and for $\ 0<p<1$ the
functional $N_{p} [\rho_{X}]$ is the \textit{exotic} $L^{p}$-norm
\cite{LuschgiPages}. For $p>1$ the $L^{p}$ spaces are reflexive, the Maxent
problem is equivalent to the minimal $L^{p}$ distance problem with
restrictions \cite{SonninoSteinbrGRE}, or to the minimal $Z_{p}(M_{X},\mu
_{X},{\rho}_{X})$. For $0<p<1$, the $L^{p}$ spaces has, in general, trivial
duals, the Maxent problem is equivalent to the maximal $L^{p}$ distance or the
maximal $Z_{p}(M_{X},\mu_{X},{\rho}_{X})$ (see \cite{SonninoSteinbrGRE}). The
case $p=1$, which corresponds to the classical Shannon entropy, is just the
border point between two radically different functional-analytic properties.
\end{remark}

\noindent The corresponding generalized entropy $S_{R,p}$, proposed by A.
R\'{e}nyi \cite{Renyi1}, and the entropy, $S_{T,p}$, proposed by C. Tsallis
\cite{Tsallis1}, \cite{Tsallis2} are given by
\begin{align}
S_{R,p}[M_{X},\mu_{X},\rho_{X}]  &  =\frac{1}{1-p}\log Z_{p}(M_{X},\mu
_{X},{\rho}_{X})\label{LL5}\\
S_{T,p}[M_{X},\mu_{X},\rho_{X}]  &  =\frac{1}{1-p}\left[  1-Z_{p}(M_{X}
,\mu_{X},{\rho}_{X})\right]  \ \label{LL7}%
\end{align}
\noindent Consider now a measure space $N=(\Omega,\mathcal{A},~n)$ with
$\sigma$-finite measure $n$, and let us denote with $P(x),$ $Q(x)$ two
probability densities:
\[
{\int\limits_{\Omega}}P(x)dn(\mathbf{x})={\int\limits_{\Omega}}
Q(x)dn(\mathbf{x})=1
\]
\noindent Note that the R\'{e}nyi divergence \cite{Renyi1},
\cite{RenyiDivergence}
\begin{equation}
D_{p}(P||Q)=\frac{1}{p-1}\log{\int\limits_{\Omega}}P^{p}Q^{1-p}dn(\mathbf{x}
)\ \label{RenyiDivergence}%
\end{equation}
\noindent is related to the R\'{e}nyi entropies, Eq.(\ref{LL5}), by
Eq.(\ref{app8}) (see Appendix~\ref{markerAppendxRenyDivergenceEntropy}). Note
that when $X$ is a finite or denumerable set, if we denote with $p_{k}$ the
probabilities of element $x_{k}\in$ $X$, the measure $\mu_{X}$ is the counting
measure on the space $X$ (equal to the number of elements in a subset), and
the family of null sets $\mathcal{N} _{X}=\{\emptyset\}$ then, from the
previous Eqs.( \ref{LL5}, \ref{LL7}, \ref{LL0.01}) we get the original
definitions from Ref.\cite{Renyi1}, \cite{Tsallis1},\cite{Tsallis2}
\begin{align}
S_{R,q}[M_{X},\mu_{X},\rho_{X}]  &  =\frac{1}{1-q}\log\sum\limits_{k}p_{k}
^{q}\\
S_{T,q}[M_{X},\mu_{X},\rho_{X}]  &  =\frac{1}{1-p}\left[  1-\sum
\limits_{k}p_{k}^{q}\right] \\
Z_{q}(M_{X},\mu_{X},{\rho}_{X})  &  =\sum\limits_{k}p_{k}^{q}%
\end{align}
\noindent Remark that, in this particular case, $S_{T,q}[M_{X},\mu_{X}
,\rho_{X}]$, as well as $Z_{q}(M_{X},\mu_{X},{\rho}_{X})$, are Lesche stable
\cite{Lesche}. Note that, from Eqs~(\ref{LL2}, \ref{LL5} and \ref{LL7}),
results
\begin{equation}
\underset{p\rightarrow1}{\lim}S_{T,q}[M_{X},\mu_{X},\rho_{X}]=\underset
{p\rightarrow1}{\lim}S_{R,q}[M_{X},\mu_{X},\rho_{X}]=S_{cl}[M_{X},\mu_{X}
,\rho_{X}] \label{LL7.0}%
\end{equation}

\subsection{Direct product of measurable spaces and the multiplicative
property of $Z_{p}[M_{X},\mu_{X},\rho_{X}]$\label{markerSubsectDirectProduct}}

In the framework of the our formalism, the multiplicative property is the
counterpart of the \textit{Postulate} 4 in the R\'{e}nyi theory\cite{Renyi1}.
In the following we overload the tensor product notation " $\otimes$ "; its
meaning results from the nature of the operand. Denote the direct product of
two measurable spaces $M_{X}=(X,\mathcal{A}_{X},\mathcal{N}_{X})$ and
$M_{Y}=(Y,\mathcal{A} _{Y},\mathcal{N}_{Y})$ by $\ M_{X}\otimes M_{Y}$,
defined as follows
\begin{equation}
M_{X}\otimes M_{X}=(X\times Y,\mathcal{A}_{X}\otimes\mathcal{A}_{Y}
,\mathcal{N}_{X\otimes Y}) \label{LL7.1}%
\end{equation}
\noindent Here $X\times Y$ is the \textit{Cartesian product} of the phase
spaces $X$ and $Y$, while the $\sigma$-algebra $\mathcal{A}_{X}\otimes
\mathcal{A}_{Y}$ is the smallest $\sigma$-algebra such that it contains all of
the elements of the Cartesian product $\mathcal{A}_{X}\times\mathcal{A}_{Y}$.
The \textit{null set ideal} $\mathcal{N}_{X\otimes Y}\mathcal{\subset A}
_{X}\otimes\mathcal{A}_{Y}$ is generated by the family $(\mathcal{A}
_{X}\otimes\mathcal{N}_{Y})\cup(\mathcal{N}_{X}\otimes\mathcal{A}_{Y})$. Note
that if $\mu_{X}\in C[M_{X}]$ and $\mu_{Y}\in C[M_{Y}]$ then their direct
product satisfies the condition $\mu_{X}\otimes\mu_{Y}\in C[M_{X}\otimes
M_{Y}]$ (we denote it also by the same symbol). The measure $\mu_{X}\otimes
\mu_{Y}\ $acting on $\mathcal{(A} _{X}\otimes\mathcal{A}_{Y})/\mathcal{N}
_{X\otimes Y}$ are defined by extension by denumerable additivity, starting
from the product subsets:
\begin{align}
(\mu_{X}\otimes\mu_{Y})(A_{X}\times A_{Y})  &  =\mu_{X}(A_{X})\mu_{Y}
(A_{Y})\label{LL7.2}\\
A_{X}  &  \in\mathcal{A}_{X};A_{Y}\in\mathcal{A}_{Y}%
\end{align}
\noindent Consider now the measures $\mu_{X}\in C(M_{X})$, $\mu_{Y}\in
C(M_{Y})$, and the densities $\rho_{X}\in L^{p}(M_{X},d\mu_{X})\cap
L^{1}(M_{X},d\mu_{X})$ and $\rho_{Y}\in L^{p}(M_{Y},d\mu_{Y})\cap L^{1}
(M_{Y},d\mu_{Y})$. The following function is also denoted with the same
symbol
\begin{equation}
\ \rho_{X}\otimes\rho_{Y}\in L^{p}(M_{X}\times M_{Y},\mu_{X}\otimes\mu
_{Y})\cap L^{1}(M_{X}\times M_{Y},\mu_{X}\otimes\mu_{Y}) \label{LL7.3}%
\end{equation}
with
\begin{align}
(\rho_{X}\otimes\rho_{Y})(x,y)  &  =\rho_{X}(x)\rho_{Y}(y)\label{LL7.4}\\
x  &  \in X;~y\in Y
\end{align}
\noindent We have the following basic proposition

\begin{proposition}
\label{markPropositionMultiplicative}Let $\rho_{X}$ , $\rho_{Y}$ are
normalized PDF
\begin{equation}
{\int\limits_{X}}\rho_{X}(x)d\mu_{X}(x\mathbf{)=}{\int\limits_{Y}}\rho
_{Y}(x)d\mu_{Y}(y\mathbf{)=1;~}\rho_{X}\geq0;~\rho_{Y}\geq0
\label{normalizationPDF}%
\end{equation}
Then we have
\begin{align}
Z_{p}[M_{X}\otimes M_{Y},\mu_{X}\otimes\mu_{Y},\rho_{X}\otimes\rho_{Y}]  &
=Z_{p}[M_{X},\mu_{X},\rho_{X}]~Z_{p}[M_{Y},\mu_{Y},\rho_{Y}]\label{LL7.5}\\
S_{R,p}[M_{X}\otimes M_{Y},\mu_{X}\otimes\mu_{Y},\rho_{X}\otimes\rho_{Y}]  &
=S_{R,p}[M_{X},\mu_{X},\rho_{X}]+S_{R,p}[M_{Y},\mu_{Y},\text{ }\rho_{Y}]
\label{LL7.6}%
\end{align}

\end{proposition}

\noindent The validity of this statement follows directly from the definitions
of the direct product, the R\'{e}nyi entropy and the functional $Z_{p}$.

\subsection{Coproduct of measurable spaces and the additivity of the
functional $Z_{p}[M_{X},\mu_{X},\rho_{X}]$\label{markerSubsectCoproduct}}

Let us study now the property encoded in the \textit{Postulate} 5' related to
the R\'{e}nyi entropy theory (Ref.~\cite{Renyi1}), transcribed in the measure
theoretic and category language and re -expressed in the term of the
functional $Z_{p}[M_{X} ,\mu_{X},\rho_{X}]$. Also in this case, we overload
the notation $\sqcup$, for the disjoint sum from the set theory. Its precise
meaning will be clear from the nature of the operands. In the following we
investigate the \textit{functorial properties}, related to \textit{Postulate}
5', of the functional $Z_{p}[M_{X},\mu_{X},\rho]$, in analogy to Proposition
\ref{markPropositionMultiplicative}. To this end we introduce the following

\begin{definition}
\label{markerDefinitionCoproduct}The coproduct of measurable spaces
$M_{X}=(X,\mathcal{A}_{X},\mathcal{N}_{X})$ and $M_{Y}=(Y,\mathcal{A}
_{Y},\mathcal{N}_{Y})$ will be denoted by $M_{X}\sqcup M_{Y}$ and have the
following structure
\begin{equation}
M_{X}\sqcup M_{Y}=\ (X\sqcup Y,\mathcal{A}_{X}\sqcup\mathcal{A}_{Y}
,\mathcal{N}_{X}\sqcup\mathcal{N}_{Y}) \label{LL7.7}%
\end{equation}
\noindent Here, $X\sqcup Y$ is the disjoint sum of the sets $X$ and $Y$, and
$\mathcal{A}_{X}\sqcup\mathcal{A}_{Y}$ is the smallest $\sigma$-algebra that
contains all of the sets of the form $A_{1}\sqcup A_{2}$, with $A_{1}
\in\mathcal{A}_{X}$ and $A_{2}\in\mathcal{A}_{Y}$, respectively. Moreover, the
new null set ideal $\mathcal{N}_{X}\sqcup\mathcal{N}_{Y}$ is the smallest
$\sigma$-algebra generated by the family $N_{1}\sqcup N_{2}$ with $N_{1}\in$
$\mathcal{N}_{X}$ and $N_{2}\in$ $\mathcal{N}_{Y}$. Let the measures $\mu
_{X}\in C(M_{X})$, $\mu_{Y}\in C(M_{Y})$ and the weights $w_{1}\geq0$,
$w_{2}\geq0$ and $w_{1}+w_{2}=1$. The measure $\mu:=w_{1}\mu_{X}\sqcup
w_{2}\mu_{Y}$ acts on the $\sigma~$- algebra $\mathcal{A}_{X}\sqcup
\mathcal{A}_{Y}$ and it is defined uniquely as the continuation by denumerable
additivity from the property
\begin{align}
\mu(A_{1})  &  =w_{1}\mu_{X}(A_{1});~A_{1}\ \in\mathcal{A}_{X}\label{LL7.8}\\
\mu(A_{2})  &  =w_{2}\mu_{Y}(A_{2});~A_{2}\in\mathcal{A}_{Y} \label{LL7.9}%
\end{align}
\noindent Let $\rho_{X}\in L^{p}(M_{X},d\mu_{X})\cap L^{1}(M_{X},d\mu_{X})$
and $\rho_{Y}\in L^{p}(M_{Y},d\mu_{Y})\cap L^{1}(M_{Y},d\mu_{Y})$\ . We define
the function $\rho:=\rho_{X}\sqcup\rho_{Y}\in L^{p}(M_{X}\sqcup M_{Y},w_{1}
\mu_{X}\sqcup w_{2}\mu_{Y})\cap L^{1}(M_{X}\sqcup M_{Y},w_{1}\mu_{X}\sqcup
w_{2}\mu_{Y})$ as follows
\begin{align*}
\rho(x)  &  =\rho_{X}(x);~if~x\in X\\
\rho(x)  &  =\rho_{Y}(x);~if~x\in Y
\end{align*}

\end{definition}

\noindent We restrict our definition of coproduct to finite terms. An example
of (denumerable infinite) coproduct is the grand canonical ensemble.

\begin{remark}
\noindent If $\rho_{X}d\mu_{X}$ and $\rho_{Y}d\mu_{Y}$ are probability
measures, then the measure $[\rho_{1}\sqcup\rho_{2}](x)[w_{1}d\mu_{X}\sqcup
w_{2}d\mu_{Y}]$ is a probability measure if $w_{1}+w_{2}=1$ .
\end{remark}

\noindent From the previous definition of the direct sum and the functional
$Z_{p} [M_{X},\mu_{X},\rho_{X}]$ the following obvious proposition results

\begin{proposition}
\label{markerPropositionAdditivityZp}The reformulation of the
\textit{Postulate} 5' (Ref.~\cite{Renyi1}) reads: the functional $Z_{p}
[M_{X},\mu_{X},\rho_{X}]$ is additive with respect to the direct sum of
measurable spaces
\begin{equation}
Z_{p}[M_{X}\sqcup M_{Y},w_{1}\mu_{X}\sqcup w_{2}\mu_{Y},\rho_{X}\sqcup\rho
_{Y}]=w_{1}Z_{p}[M_{X},\mu_{X},\rho_{X}]+w_{2}Z_{p}[M_{Y},\mu_{Y},\rho_{Y}]
\label{LL7.10}%
\end{equation}

\end{proposition}

\subsection{ Universal properties of the direct product and direct sum in the
category of measurable spaces \label{markerSubsectUniversalProperties}}

In the following we prove that the basic binary operations on measurable
spaces, the direct product and the direct sum, defined in the previous
section, have universality properties in the category of measurable spaces
$MES$.

\noindent\ Consider the direct product $M=M_{X}\otimes M_{Y}$ of measurable
spaces $M_{X}=(X,\mathcal{A}_{X},\ \mathcal{N}_{X})$ and $M_{Y}=(Y,\mathcal{A}
_{Y},\mathcal{N}_{Y})$. Observe that the canonical projections $p_{X}:X\times
Y\rightarrow X$, $p_{Y}:X\times Y\rightarrow Y$, are measurable and induce the
morphisms $\pi_{X}:M_{X}\otimes M_{Y}\rightarrow M_{X}$ and $\pi_{Y}
:M_{X}\otimes M_{Y}\rightarrow M_{Y}$ \ \ between the objects of $MES$. We
have the following

\begin{proposition}
\label{markerPropositionProiectionsAreMorphism} In the category $MES$ the
applications $\pi_{X}:M_{X}\otimes M_{Y}\rightarrow M_{X}$, $\pi_{Y}
:M_{X}\otimes M_{Y}\rightarrow M_{Y}$, which are naturally induced by
canonical projections $p_{X}:X\times Y\rightarrow X$ and $p_{Y}:X\times
Y\rightarrow Y$, are morphisms.

\begin{proof}
\noindent The measurability of $\pi_{X}$ is direct consequence of the fact
that the canonical projection maps are measurable, in fact the measurability
of the canonical projections is an alternative definition of the product of
$\sigma$ algebras. The nonsingularity property \ $p_{X}^{(-1)} (\mathcal{N}%
_{X})\subset\mathcal{N}_{X\times Y}$ \ results \ from $p_{X}^{(-1)}%
(\mathcal{N}_{X})=\mathcal{N}_{X}\times\mathcal{A}_{Y} \subset\mathcal{N}%
_{X\times Y}$ .
\end{proof}
\end{proposition}

From the previous Proposition \ref{markerPropositionProiectionsAreMorphism}
results immediately the following Theorem

\begin{theorem}
\label{markerTheoremDirectProd} In the category $MES$, the direct product has
the universal property. Let $M_{X}=(X,\mathcal{A}_{X},\mathcal{N}_{X})$.
$M_{Y}=(Y,\mathcal{A}_{Y},\mathcal{N}_{Y})$ and $M=(Z,\mathcal{A}
_{Z},\mathcal{N}_{Z})$ measurable spaces that are objects of the category
$MES$, such that there exists morphisms $\phi_{X}\in Hom(M,M\,_{X})$ and
$\phi_{Y}\in Hom(M,M\,_{Y})$. Then there exists an unique morphism $\theta\in
Hom(M,M_{X}\otimes M_{Y})$ such that
\begin{align}
\phi_{X}  &  =\pi_{X}\circ\theta\label{LL8}\\
\phi_{Y}  &  =\pi_{Y}\circ\theta\label{LL9}%
\end{align}
\noindent where $\pi_{X}$ , $\pi_{Y}$ are the morphism defined in Proposition
\ref{markerPropositionProiectionsAreMorphism} .

\begin{proof}
\noindent The morphism $\theta$ is induced by the application $T:Z\rightarrow
X\times Y$ defined as $Z\ni z\rightarrow T(z):=(\phi_{X}(z),\phi_{Y}(z))\in
X\times Y$. and it is unique. In order to prove that $\theta$ is a morphism we
have to prove that $T$ is measurable and it is nonsingular. To prove that
$T:Z\rightarrow X\times Y$ is measurable, we recall that it is sufficient to
prove that, for all $A\in\mathcal{A}_{X}$, $B\in\mathcal{A}_{Y}$, we have the
property $T^{(-1)}(A\times B)\in\mathcal{A}_{Z}$, a property resulting from
the measurability of $\phi_{X}$ and $\phi_{Y}$. Note that to prove the
inclusion $T^{-1}(\mathcal{N}_{X\times Y})\subset\mathcal{N}_{Z}$, it is
sufficient to demonstrate for the generating subsets $T^{-1}(\mathcal{N}%
_{X}\times\mathcal{A}_{Y})\subset\mathcal{N}_{Z}$ (which follows from the
nonsingularity of $\phi_{X}$ and $T^{-1}(\mathcal{A} _{X}\times\mathcal{N}%
_{Y})\subset\mathcal{N}_{Z}$) that this is the consequence of the
nonsingularity of $\phi_{Y}$.
\end{proof}
\end{theorem}

\noindent In conclusion the direct product operation has the natural
functorial property, so the multiplicative property Eq.~(\ref{LL7.5}) of the
functional $Z_{p}(M_{X},\mu_{X},\rho_{X})$ appears as an algebraic
compatibility property. By simple reversal of the arrows, we are lead to the
corresponding universality property of the coproduct in the category $MES$. We
have the following obvious proposition

\begin{proposition}
\label{markerPropositionCanonicalInjection} In the category $MES$, consider
the objects $M_{X},M_{Y}$. The applications $\iota_{X}:M_{X}\rightarrow
M_{X}\sqcup M_{Y}$ and $\iota_{Y}:M_{Y}\rightarrow M_{X}\sqcup M_{Y}$, induced
naturally by the canonical injections $i_{X}:X\rightarrow X\sqcup Y$,
$i_{Y}:Y\rightarrow X\sqcup Y$, are morphism in the category $MES$.
\end{proposition}

\begin{proof}
The injections $i_{X}$, $i_{Y}$ are measurable. Suppose that $N_{1}\sqcup
N_{2}\in\mathcal{N}_{X}\sqcup\mathcal{N}_{Y}$, with $N_{1}\in\mathcal{N}_{X}
$, $N_{2}\in\mathcal{N}_{Y}$ (see \textit{Definition}%
~\ref{markerDefinitionCoproduct}). Then, $i_{X}^{(-1)}(N_{1}\sqcup
N_{2})=N_{1}$, $i_{Y}^{(-1)}(N_{1}\sqcup N_{2})=N_{2}$, so $i_{X}$ and $i_{Y}$
are nonsingular, which completes the proof that $\iota_{X}$ , $\iota_{Y}$ are
morphisms in the category $MES$.
\end{proof}

\noindent By reversing the arrows, in analogy to the Theorem
\ref{markerTheoremDirectProd}, we obtain the following result

\begin{theorem}
\label{markerTheoremCanonicalInjection} In the category $MES$ the direct sum
of the objects has the following universality property. Let denote with
$M_{X}=(X,\mathcal{A}_{X},\mathcal{N}_{X})$, $M_{Y}=(Y,\mathcal{A}
_{Y},\mathcal{N}_{Y})$ and $M=(Z,\mathcal{A}_{Z},\mathcal{N}_{Z})$ measurable
spaces that are objects of the category $MES$, such that there exists
morphisms $\phi_{X}\in Hom(M_{X},M)$ and $\phi_{Y}\in Hom(M_{Y},M)$. Then,
there exists an unique morphism $\gamma\in Hom(M_{X}\sqcup M_{Y},M)$ such
that
\begin{align*}
\gamma\circ\iota_{X}  &  =\phi_{X}\\
\gamma\circ\iota_{Y}  &  =\phi_{Y}%
\end{align*}
\noindent where $\iota_{X}$ ,$\iota_{X}$ are the morphisms defined in
Proposition~\ref{markerPropositionCanonicalInjection}.

\begin{proof}
\noindent The morphism $\gamma$ is induced by the map $g:X\sqcup Y\rightarrow
Z$ defined as follows. If $x\in X$ then $g(x):=\phi_{X}(x)\in Z$, and in the
case $x\in Y$, then $g(x):=\phi_{Y}(x)\in Z$. The measurability of the map $g$
results from the measurability of $\phi_{X}$ and $\phi_{Y}$. The inclusion
$g^{(-1)}(\mathcal{N}_{Z})\subset\mathcal{N}_{X}\sqcup\mathcal{N}_{Y}$ results
from the nonsingularity of $\phi_{X}$ and $\phi_{Y}$.
\end{proof}
\end{theorem}

\noindent In conclusion, the direct sum operation has natural category
theoretic properties. He,Hence, the additivity property Eq.( \ref{LL7.10}) of
the functional $Z_{p}(M_{X},\mu_{X},\rho_{X})$ is not an artificial construction.

\subsection{The monoidal categories associated to product and coproduct
\label{markerSubsectMonoidalCategory}}

We recall the following

\begin{proposition}
\label{markerTheorDirProdMonoidal}\cite{MacLaneCategory} Let $\mathcal{C}$ be
a category such that for all objects $A,B\in Ob(\mathcal{C)}$ exists their
direct product $A\otimes B$, having the universal property. Then, there exists
a \emph{covariant functor} $F$ from the product category to $\mathcal{C}$,
$\mathcal{C\times}$ $\mathcal{C\rightarrow}$ $\mathcal{C}$ ,defined as
follows. For the object $(A,B)$ of $\mathcal{C\times}\mathcal{C}$, where $A,B$
are objects of $\mathcal{C}$, we have
\[
F((A,B)):=A\otimes B
\]
\noindent For the pair of morphisms $(u,v)\in Hom((A,B),\ (A^{\prime
},B^{\prime}))$ with $u\in Hom(A,A^{\prime})$, $v\in Hom(B,B^{\prime})$, from
the category $\mathcal{C}\times\mathcal{C}$ there exists an unique morphism
$w$ in the category $\mathcal{C}$, $w\in Hom(A\otimes B,A^{\prime}\otimes
B^{\prime})$ uniquely fixed by the conditions
\begin{align*}
w  &  =F((u,v))\\
p_{A^{\prime}}\circ w  &  =u\circ p_{A}\\
p_{B^{\prime}}\circ w  &  =v\circ p_{A}%
\end{align*}
\noindent We denoted with $p_{A}$, $p_{B}$ the projections from $Hom(A\otimes
B,A)$, $Hom(A\otimes B,B)$, and $p_{A^{\prime}}$ are $p_{B^{\prime}}$ the
projections from $Hom(A^{\prime}\otimes B^{\prime},A^{\prime})$,
$Hom(A^{\prime}\otimes B^{\prime},B^{\prime})$. The map $(u,v)\rightarrow
F((u,v))$ has the functorial property.

\noindent Let $(u,v)\in Hom((A,B),(A^{\prime},B^{\prime}))$ and $(u^{\prime
},v^{\prime})\in Hom((A^{\prime},B^{\prime}),(A",B"))$. Then,
\[
F((u^{\prime}\circ u,v^{\prime}\circ v~\ ))=F((u^{\prime},v^{\prime}~\ ))\circ
F((u,v~\ ))\in Hom(A\otimes B,A"\otimes B")
\]
\noindent If in the category $\mathcal{C}$ we have an unit object, then
$\mathcal{C}$ is a monoidal category.
\end{proposition}

\noindent Similarly, by duality arguments, we have the following result for
the direct sum (coproduct)

\begin{proposition}
\label{markerPropDirectSum}\cite{MacLaneCategory} Let $\mathcal{C}$ be a
category such that for all objects $A,B$ from $Ob($ $\mathcal{C)}$ exists
their direct sum $A\sqcup B$, having the universal property. Then, there
exists a \ \emph{covariant functor} $G$ from the product category
$\mathcal{C\times}\mathcal{C\rightarrow}\mathcal{C}$ defined as follows. For
the object $(A,B)$ of $\mathcal{C\times}$ $\mathcal{C}$, where $A,B$ are
objects of $\mathcal{C}$ we have
\[
G((A,B)):=A\sqcup B
\]
\noindent For the pair of morphisms $(u,v)\in Hom((A,B),(A^{\prime},B^{\prime
}))$, with $u\in Hom(A,A^{\prime})$ and $v\in Hom(B,B^{\prime})$, from the
category $\mathcal{C}\times\mathcal{C}$ there exists an unique morphism $w$ in
the category $\mathcal{C}$, $w\in Hom(A\sqcup B,A^{\prime}\sqcup B^{\prime})$
uniquely fixed by the conditions
\begin{align*}
w  &  =G((u,v))\\
w\circ i_{A}  &  =i_{A^{\prime}}\circ u\\
w\circ i_{B}  &  =i_{B^{\prime}}\circ v
\end{align*}
\noindent We denoted with $i_{A}$ , $i_{B}$ the canonical injections from
$Hom(A,A\sqcup B)$, $Hom(B,A\sqcup B)$, and with $i_{A^{\prime}}$ ,
$i_{B^{\prime}}$ the injections from $Hom(A^{\prime},A^{\prime}\sqcup
B^{\prime})$, $Hom(B^{\prime},A^{\prime}\sqcup B^{\prime})$. The association
$(u,v)\rightarrow G((u,v))$ has the functorial property. Let $(u,v)\in
Hom((A,B),(A^{\prime},B^{\prime}))$ and $(u^{\prime},v^{\prime})\in
Hom((A^{\prime},B^{\prime}),(A",B"))$ then,
\[
G((u^{\prime}\circ u,v^{\prime}\circ v~\ ))=G((u^{\prime},v^{\prime}~\ ))\circ
G((u,v~\ ))\in Hom(A\sqcup B,A"\sqcup B")
\]
If in the category $\mathcal{C}$ we have a null object then, $\mathcal{C}$ is
a monoidal category with respect to direct sum.
\end{proposition}

\noindent We emphasize that, despite the fact that the construction of the
direct sum is dual to the direct product, from the previous proposition
(\ref{markerPropDirectSum}) the functor $G$ is a covariant functor. In the
category $MES$ we have an unit object as well as the null object. The unit
object is denoted with $M_{1}:=(1,\mathcal{A}_{1},\mathcal{N}_{1})$, where $1$
is the one point set \cite{MacLaneCategory}, $\mathcal{A}_{1}$ is the trivial
$\sigma$-algebra consisting in the one point set $1$, $\varnothing$, and
$\mathcal{N}_{1}=\{\varnothing\}$, respectively. The (more or less formal)
null object $M_{0}$, with respect to the direct sum, is the object generated
by the empty set $M_{0}:=(\varnothing,\mathcal{A}_{\varnothing},\ \mathcal{N}%
_{\varnothing})$. So we have the following

\begin{conclusion}
The category $MES$ is a monoidal category both with respect to the product
$\otimes$ and the coproduct $\sqcup$.
\end{conclusion}

\section{ Axioms \label{markerSectionAxioms}}

We expose another approach, based on category theory, to the problem of the
naturalness of the choice of the family of functions $g_{\alpha}$ used in the
definition of the entropy \cite{Renyi1}. We prove that this problem may be
treated if we take into account the additivity and the multiplicative
properties of the functional $Z_{p}^{{}}$. We mention that a possible
candidate for the generalization of the symmetry \textit{Postulate} 1
\cite{Renyi1} is the requirement of invariance of the generalized entropy
under measure preserving transformations. Recall that the group generated by
finite permutations is the maximal \emph{measure preserving group} with
respect to the counting measure. The problem is that there are plenty of
measures such that the measure preserving group is trivial (for instance, the
atomic measure for 2 element set with $\mu(1)\neq\mu(2)$). To avoid this
problem, we observe that \emph{Postulate 1} and \textit{Postulate} 5' in the
original R{\'{e}}nyi theory \cite{Renyi1} can be generalized as follows. For a
given measurable function $f(x)$ on the measured space $M_{X}$ and $\mu\in
C(M_{X})$, let us define
\begin{equation}
m_{f}(M_{X},\mu,t)=\mu\left[  \left(  x|x\in X~\&~f(x)\leq t\right)  \right]
\label{mm1}%
\end{equation}
\noindent Note that $m_{f}(M_{X},\mu,t)$ is invariant under measure preserving
transformations. In addition
\begin{equation}
Z_{p}[M_{X},\mu_{X},\rho]~=\int\limits_{0}^{\infty}t^{p}dm_{\rho}(M_{X}
,\mu,t) \label{mm2}%
\end{equation}
\noindent Then, the \textit{Postulate} 1 (the symmetry property) and
\textit{Postulate} 5' (the additivity property expressed in Proposition
\ref{markerPropositionAdditivityZp}) can be generalized as follows.
\textit{Postulate} 1 \& \textit{Postulate} 5'
\begin{align}
Z_{p}[M_{X},\mu_{X},\rho_{X}]  &  =\int\limits_{0}^{\infty}h_{X}%
(t)~dm_{\rho_{X}}(M_{X},\mu,t)={\int\limits_{X}h_{X}\ }\left[  \rho
_{X}(x)\right]  d\mu_{X}(x\mathbf{)}\label{mm3}\\
h_{X}(x)  &  >0;~if~~x>0 \label{mm3.1}%
\end{align}
\noindent for some Borel measurable function $h_{X}(t)$ with
\begin{equation}
h_{X}(0)=0 \label{mm3.5}%
\end{equation}
\noindent The last requirement result by considering the case when the support
of $\rho_{X}$ is concentrated on a proper subset of $X$ and by using
Eq.(\ref{LL7.10}) . The generalization of the \textit{Postulate} 2 (the
continuity property) is straightforward. Be $h(x)$ continuos and $\rho_{X}\in
L^{1}(M_{X},\mu_{X})$, we get
\begin{equation}
{h}_{X}\left[  \rho(x)\right]  \in L^{1}(M_{X},\mu_{X}) \label{mm4}%
\end{equation}
\noindent In our settings, the analog of the \textit{Postulate} 4 (the
additivity property) \cite{Renyi1} is the multiplicative property given by
Eq.(\ref{LL7.5}) and Proposition \ref{markPropositionMultiplicative}. By using
Eqs.(\ref{LL7.5}, \ref{mm3}, \ref{mm3.5} and \ref{mm4}), and by continuity of
the functions $h_{XY}$, $h_{X}$, $h_{Y}$ for all $x,y\geq0$, we obtain the
following functional equation (valid almost everywhere)
\begin{equation}
h_{XY}(x~y)=h_{X}(x)h_{Y}(y);~x,y\in\mathbb{R} \label{mm5}%
\end{equation}
\noindent By arguments similar to the proof of the \textit{uniqueness}, from
Theorem 2 \cite{Renyi1}), we get Eq.~(\ref{mm2}) (for details see the
Appendix~ \ref{markerAppendixSubsectFunctionalEquaton}): there exists an
universal family of functions, independent of $X$ , parametrized by the
positive parameter $p$ such that
\begin{align}
h_{X}(x)  &  =x^{p}C_{X}\label{mm6}\\
h_{Y}(y)  &  =y^{p}C_{Y}\label{mm6.1}\\
h_{XY}(z)  &  =z^{p}C_{X}C_{Y} \label{mm6.2}%
\end{align}

\section{Summary and conclusions}

We proved that the most natural setting for treating the axiomatic approach to
the study of definitions of measures of information or uncertainty, is the
formalism of the category theory, that was invented for the most difficult,
apparently contradictory aspects of the foundation of mathematics. In this
respect we introduced a category of measurable spaces $MES$. We proved that in
the category $MES$ exist the direct product and the direct sum, having
universal properties. We proved that the functional $Z_{p}(M_{X},\mu_{X}
,{\rho}_{X})$ defined in Eq.(\ref{LL0.01}), which appears in the definition of
both R\'{e}nyi and Tsallis entropies, has algebraic compatibility properties
with respect to direct product and direct sum, as shown in Eqs~(\ref{LL7.5})
and (\ref{LL7.10} ).

\noindent The main conclusions may be summarized as follows.

\noindent\textbf{(1)} The natural measure of the quantity of information is
the family of functionals $Z_{p}(M_{X},\mu_{X},{\rho}_{X})$ given by
Eq.~(\ref{LL0.01}), (defined in the Fr\'{e}chet space for $0<p<1$, and in the
Banach space for $p>1$), and the classical Shannon entropy by Eq.~(\ref{LL2});

\noindent\textbf{(2)}The category $MES$ is the natural framework for treating
the problems related to the measure of the information, in particular in
reformulating the R\'{e}nyi axioms;

\noindent\textbf{(3)} The category $MES$ is a monoidal category with respect
to direct product and coproduct and the functional $Z_{p}(M_{X},\mu_{X},{\rho
}_{X}$ $)$ has natural \textit{compatibility properties} with respect to the
product (it is multiplicative) and the coproduct (it is additive);

\noindent\textbf{(4)} Up to a multiplicative constant, it is possible to
recover the exact form of the functional $Z_{p}(M_{X},\mu_{X},{\rho}_{X})$
defining the generalized entropies from a system of axioms that generalize the
ones adopted by R\'{e}nyi \cite{Renyi1}.

\begin{acknowledgement}
The authors are grateful to Prof. M. Van Schoor and Dr D. Van Eester from
Royal Military School, Brussels. Gy\"{o}rgy Steinbrecher is indepted to Prof.
C. P. Niculescu from Mathematics Department, University of Craiova, Romania,
and S. Barasch for discussions on category theory. Giorgio Sonnino is also
grateful to Prof. P. Nardone and Dr P. Peeters of the Universit\'{e} Libre de
Bruxelles (ULB) for useful discussions and suggestions.
\end{acknowledgement}

\section{Appendix}

\subsection{R\'{e}nyi divergence and entropy
\label{markerAppendxRenyDivergenceEntropy}}

Suppose to have a measurable space~$(\Omega,\mathcal{A},m)$ with a finite or
$\sigma$ -finite measure $\mu$ and a normalized PDF $\rho(x)$, i.e.
$\int\limits_{\Omega}\rho(x)d\mu(x)=1$. Only in this Subsection we adopt the
following definitions
\begin{align}
U(\rho,d\mu,\alpha)  &  :={\int\limits_{\Omega}}\left[  \rho(x)\right]
^{\alpha}d\mu(x)\label{app1}\\
S_{R,\alpha}(\rho,d\mu)  &  =\frac{1}{1-\alpha}\log U(\rho,d\mu,\alpha)
\label{app2}%
\end{align}
\noindent Consider now a measurable space $N=(\Omega,\mathcal{A},~n)$ with
$\sigma$-finite measure $n$. We also denote with $P(x),$ $Q(x)$ two
probability densities, satisfying the condition
\begin{equation}
{\int\limits_{\Omega}}P(x)dn(x)={\int\limits_{\Omega}}Q(x)dn(x)=1 \label{app3}%
\end{equation}
\noindent The R\'{e}nyi divergence reads
\begin{equation}
D_{p}(P||Q)=\frac{1}{p-1}\log{\int\limits_{\Omega}}P^{p}Q^{1-p}
dn(x)\ \label{app4}%
\end{equation}
\noindent According to the notation Eq.(\ref{app1}) and normalization
Eq.(\ref{app3}), we get
\begin{equation}
{\int\limits_{\Omega}}Q^{1-p}dn(x)=U(Q,dn,1-p) \label{app5}%
\end{equation}
and $P_{1}(x):=P(x)/U(Q,dn,1-p)$ is normalized with respect to the measure
$Q^{1-p}dn(\mathbf{x})$. Consequently, from Eq.(\ref{app1}), we find
\begin{align}
{\int\limits_{\Omega}}P^{p}Q^{1-p}dn(x)\  &  =U\left[  \frac{P}{Z}
,Q^{1-p}dn,p\right]  Z^{p}\label{app6}\\
Z  &  :=U(Q,dn,1-p) \label{app7}%
\end{align}
\noindent By using the notation Eqs(\ref{app2}, \ref{app4}-\ref{app7}) we
obtain the following relation between the R\'{e}nyi entropies and the
divergences
\begin{equation}
D_{p}(P||Q)=-S_{R,p}\left[  \frac{P}{Z},Q^{1-p}dn\right]  +\frac{p^{2}}
{p-1}S_{R,1-p}\left[  Q,dn\right]  \label{app8}%
\end{equation}

\subsection{ Solution of the functional equation Eq.(\ref{mm5}
)\label{markerAppendixSubsectFunctionalEquaton}}

Using Eq.~(\ref{mm3.1}) with $\rho\geq0$, we note that we can use the double
logarithmic scale by performing the following change of variables
\begin{align}
f_{X}(u)  &  =\log h_{X}(\exp(u))\ \label{app8.1}\\
f_{Y}(v)  &  =\log h_{Y}(\exp(v))\ \label{app8.2}\\
f_{XY}(z)  &  =\log h_{XY}(\exp(z))\ \label{app8.3}%
\end{align}
\noindent Hence, Eq.~(\ref{mm5}) reads
\begin{equation}
f_{XY}(u+v)=f_{X}(u)+f_{Y}(v) \label{app9}%
\end{equation}
\noindent From Eq.~(\ref{app9}), we obtain
\begin{align}
f_{XY}(u+v)-f_{XY}(v)  &  =f_{Y}(v)-f_{Y}(0)\\
\lbrack f_{XY}(u+v)-f_{XY}(v)]-[f_{XY}(u+0)-f_{XY}(0)]  &  =0\\
f_{XY}(u+v)-f_{XY}(v)-f_{XY}(u)  &  =f_{XY}(0) \label{101}%
\end{align}
\noindent The Eq.(\ref{app13}) admits the particular constant solution
\begin{equation}
F_{part}(z)\equiv f_{XY}(0) \label{app13}%
\end{equation}
\noindent The solution of corresponding homogenous equation
\begin{equation}
f_{XY}(u+v)-f_{XY}(v)-f_{XY}(u)=0 \label{app14}%
\end{equation}
\noindent may be found by using again the continuity of the function
$h_{XY}(\rho)$ (See also \cite{DunfordSchw} I.3.1, page 8), i.e.,
\begin{equation}
f_{XY}(z)=pz\text{ } \label{app15}%
\end{equation}
\noindent The general solution reads
\[
f_{XY}(z)=f_{XY}(0)+pz
\]
\noindent By using Eq.~(\ref{app9}), we get the universal linear slope $p$
\begin{align*}
f_{X}(u)  &  =f_{X}(0)+pu\\
f_{Y}(v)  &  =f_{Y}(0)+pv\\
f_{XY}(0)  &  =f_{X}(0)+f_{Y}(0)
\end{align*}
\noindent and, by Eqs~(\ref{app8.1}-\ref{app8.3}), up to undetermined
multiplicative constants $C_{X}\!=\!\exp(f_{X}(0))$, $C_{Y}\!=\!\exp
(f_{Y}(0))$, we find Eqs~(\ref{mm6}-\ref{mm6.2}).

\end{document}